\documentclass[10pt]{article}

\usepackage{amsmath}
\usepackage{amsfonts}
\usepackage{amsthm}
\usepackage[english]{babel}

\newtheorem{theorem}{Theorem}
\newtheorem{lemma}{Lemma}
\newtheorem{definition}{Definition}
\newtheorem{corollary}{Corollary}
\newtheorem{proposition}{Proposition}
\newtheorem{remark}{Remark}

\newcommand{\nd}{{\rm and }}

\newcommand{\mR}{\mathbb{R}}
\newcommand{\mC}{\mathbb{C}}
\newcommand{\mN}{\mathbb{N}}
\newcommand{\mE}{\mathbb{E}}

\newcommand{\mS}{\mathbb{S}}

\newcommand{\mB}{\mathbb{B}}

\newcommand{\cM}{\mathcal{M}}
\newcommand{\cH}{\mathcal{H}}

\newcommand{\cF}{\mathcal{F}}
\newcommand{\cP}{\mathcal{P}}
\newcommand{\cC}{\mathcal{C}}
\newcommand{\cR}{\mathcal{R}}

\newcommand{\cS}{\mathcal{S}}

\newcommand{\cW}{\mathcal{W}}

\newcommand{\ux}{\underline{x}}
\newcommand{\uxb}{\underline{x} \grave{}}
\newcommand{\uxi}{\underline{\xi}}
\newcommand{\uyb}{\underline{y} \grave{}}
\newcommand{\uy}{\underline{y}}

\newcommand{\ds}{d \sigma_{\ux}}

\newcommand{\pj}{\partial_{x_j}}
\newcommand{\pjb}{\partial_{{x \grave{}}_{j}}}

\newcommand{\px}{\partial_x}

\newcommand{\upx}{\partial_{\underline{x}}}
\newcommand{\upxb}{\partial_{\underline{{x \grave{}}} }}

\begin{document}
%%%%%%%%%%%%%%%%%%%%%%%%%%%%%%%%%%%%%%%%%%%%%%%%%%%%%%%%
\title{Integration in superspace using distribution theory}

\author{K.\ Coulembier\thanks{Ph.D. Fellow of the Research Foundation - Flanders (FWO), E-mail: {\tt Coulembier@cage.ugent.be}}, H.\ De Bie\thanks{Corresponding author} \thanks{Ph.D. Fellow of the Research Foundation - Flanders (FWO), E-mail: {\tt Hendrik.DeBie@UGent.be}} , F.\ Sommen\thanks{E-mail: {\tt fs@cage.ugent.be}}}

\date{\small{Clifford Research Group -- Department of Mathematical Analysis}\\
\small{Faculty of Engineering -- Ghent University\\ Krijgslaan 281, 9000 Gent,
Belgium}}

\maketitle

\begin{abstract}
In this paper, a new class of Cauchy integral formulae in superspace is obtained, using formal expansions of distributions. This allows to solve five open problems in the study of harmonic and Clifford analysis in superspace.
\end{abstract}

\textbf{MSC 2000 :}   58C50, 30G35, 26B20\\
\noindent
\textbf{Keywords :}   Cauchy formula, superspace, supersphere, Clifford analysis
%%%%%%%%%%%%%%%%%%%%%%%%%%%%%%%%%%%%%%%%%%%%%%%%%%%%%%%%%%%%%%%%%%%%%%%%%%%%%%%%%%%%%%%%

\newpage
\section{Introduction}

In a previous set of papers (see a.o. \cite{DBS1, DBS2, DBS4, 
DBS5, DBE1, DBS9}) we have developed a theory of harmonic analysis and 
Clifford analysis in superspace. Superspaces are spaces equipped with 
both a set of commuting variables and a set of anti-commuting 
variables (generating a so-called Grassmann algebra). They are usually 
studied from the point of view of algebraic or differential geometry 
(see \cite{MR732126, MR565567, MR778559, MR574696}). Our approach, on the other hand, was 
based on a generalization of harmonic and Clifford analysis by introducing a set of differential operators (such as 
a Dirac and Laplace operator) and on the study of the special 
functions and orthogonal polynomials related to these operators.

The aim of the present paper is to solve several open questions that 
have arisen during our previous research. It turns out that these 
problems can be solved by using a distributional approach to 
integration in superspace.

First of all, let us briefly discuss the problems that will be 
answered in the present paper.

In \cite{DBS5}, we introduced an integral over the supersphere of a 
polynomial $R$ using an old result of Pizzetti (see \cite{PIZZETTI}) 
as follows:
\begin{equation}
\label{Pizzetti1}
\int_{SS} R = \sum_{k=0}^{\infty} (-1)^k \frac{2 \pi^{M/2}}{2^{2k} 
k!\Gamma(k+M/2)} (\Delta^{k} R )(0),
\end{equation}
where $\Delta$ is the super Laplace operator and $M$ the associated 
superdimension.
However, with this approach we could only integrate polynomials and 
not more general superfunctions. In principle, one could use the 
Hahn-Banach theorem and the density of polynomials in several types of 
function spaces to extend the Pizzetti formula, but this would require 
good estimates of derivatives, which are of course hard to obtain.
A second problem is that we could only prove the uniqueness of the 
integral (\ref{Pizzetti1}) in the case when the superdimension $M$ is 
not even and negative (see \cite{DBE1}).

Furthermore, in \cite{DBSCauchy} we obtained a Cauchy integral formula 
in superspace and several related results, leading to nice 
generalizations of well-known facts in complex analysis. Again, two 
problems could not be settled. We did not obtain a generalization of 
Morera's theorem, stating that if the integral of a function over 
every contour in an open set $\Omega$ is zero, then this function is 
monogenic in $\Omega$. Secondly, although we did obtain a 
Cauchy-Pompeiu formula, this formula does not allow to reconstruct a 
monogenic function completely. Apart from these two major problems, we would also like to connect the Green formula, connecting integration over the supersphere with integration over the superball (see \cite{DBS5}) with more general Cauchy integral formulae.

Finally, in \cite{DBS9}, we introduced a Fourier transform in 
superspace. As an application we defined a Radon transform in 
superspace by means of two consecutive Fourier transforms. We were at 
that time not able to find an interpretation of this transform as an 
integral over the set of all hyperplanes.

Summarizing, we want to solve the following 5 problems:
\begin{itemize}
\item \textbf{P1}: find a closed formula for the Pizzetti integral 
over the supersphere
\item \textbf{P2}: prove uniqueness of the Pizzetti integral in the 
case $M \in -2 \mN$
\item \textbf{P3}: obtain Morera's theorem for monogenic functions in 
superspace
\item \textbf{P4}: obtain a suitable Cauchy-Pompeiu formula in superspace
\item \textbf{P5}: give an interpretation of the super Radon transform.
\end{itemize}
All these problems can be solved using a similar technique. Let us 
briefly sketch the main idea in the case we are dealing with 
$\mR^{m}$. Suppose that $\cM_{1}$ is a $m-1$ dimensional manifold 
contained in $\mR^{m}$ and determined by an equation $\nu(x_{1}, 
\ldots, x_{m}) =0$. Then integration over this manifold can be rewritten 
as follows:
\begin{equation}
\int_{\cM_{1}} . = \int_{\mR^{m}} \delta(\nu(x_{1},\ldots, x_{m})) (.) dV(\ux)
\label{intrand}
\end{equation}
with $\delta$ the Dirac distribution and $dV(\ux)$ the classical 
Lebesgue measure.

Similarly, integration over an $m$-dimensional manifold $\cM_{2}$ in 
$\mR^{m}$, determined by some inequality $\nu(x_{1},\ldots, x_{m})>0$ 
can be rewritten as
\begin{equation}
\int_{\cM_{2}} . = \int_{\mR^{m}} H(\nu(x_{1},\ldots, x_{m}))(.) dV(\ux)
\label{intinhoud}
\end{equation}
where $H$ is the Heaviside function.

Problems \textbf{P1} and \textbf{P5} will be solved by 
generalizing equations (\ref{intrand}) and (\ref{intinhoud}) to 
superspace, where the distributions $\delta$ and $H$ will be developed 
in formal Taylor series. Problem \textbf{P2} is solved using different means, but we have drawn inspiration from our work on problem \textbf{P1}. Problems \textbf{P3} and \textbf{P4} will be 
solved by reformulating the Cauchy integral formula in $\mR^{m}$ (a 
well-known result in Clifford analysis, see e.g. \cite{MR697564, MR1169463}) in terms of distributions, 
which will allow us to obtain a suitable superspace extension.

The paper is organized as follows. In section 2 we repeat some results on the theory of harmonic and Clifford analysis in superspace, necessary for the sequel. In section 3 we obtain a new class of Cauchy formulas in superspace. In section 4 we obtain Morera's theorem and we state a new Cauchy-Pompeiu formula, hence solving \textbf{P3} and \textbf{P4}. In section 5 we discuss integration over the supersphere and solve problems \textbf{P1} and \textbf{P2}. Finally, in section 6 we give a new definition for the Radon transform in superspace and connect it with the previous definition, thus solving \textbf{P5}.

\section{Preliminaries}
\label{preliminaries}

Superspaces are spaces where one considers not only commuting (bosonic) but also anti-commuting (fermionic) co-ordinates (see a.o. \cite{MR732126,MR565567}). In our approach to superspace (see \cite{DBS1}), we start with the real algebra $\cP \otimes\cC=\mbox{Alg}(x_i,{x \grave{}}_j)\otimes  \mbox{Alg}(e_i,{e \grave{}}_j)= \mbox{Alg}(x_i, e_i; {x \grave{}}_j,{e \grave{}}_j)$, $i=1,\ldots,m$, $j=1,\ldots,2n$
generated by
\begin{itemize}
\item $m$ commuting variables $x_i$ and $m$ orthogonal Clifford generators $e_i$
\item $2n$ anti-commuting variables ${x \grave{}}_i$ and $2n$ symplectic Clifford generators ${e \grave{}}_i$
\end{itemize}
subject to the multiplication relations
\[ \left \{
\begin{array}{l} 
x_i x_j =  x_j x_i\\
{x \grave{}}_i {x \grave{}}_j =  - {x \grave{}}_j {x \grave{}}_i\\
x_i {x \grave{}}_j =  {x \grave{}}_j x_i\\
\end{array} \right .
\quad \mbox{and} \quad
\left \{ \begin{array}{l}
e_j e_k + e_k e_j = -2 \delta_{jk}\\
{e \grave{}}_{2j} {e \grave{}}_{2k} -{e \grave{}}_{2k} {e \grave{}}_{2j}=0\\
{e \grave{}}_{2j-1} {e \grave{}}_{2k-1} -{e \grave{}}_{2k-1} {e \grave{}}_{2j-1}=0\\
{e \grave{}}_{2j-1} {e \grave{}}_{2k} -{e \grave{}}_{2k} {e \grave{}}_{2j-1}=\delta_{jk}\\
e_j {e \grave{}}_{k} +{e \grave{}}_{k} e_j = 0\\
\end{array} \right .
\]
and where moreover all elements $e_i$, ${e \grave{}}_j$ commute with all elements $x_i$, ${x \grave{}}_j$. The algebra generated by all generators $e_i$, ${e \grave{}}_j$ is denoted by $\cC$. In the case where $n = 0$ we have that $\cC \cong \mR_{0,m}$, the standard orthogonal Clifford algebra with signature $(-1,\ldots,-1)$. When $m=0$, we have that $\cP\otimes\cC=\Lambda_{2n}\otimes \cW_{2n}$, with $\Lambda_{2n}$ the Grassmann algebra generated by the  ${x \grave{}}_i$ and $\cW_{2n}$ the Weyl algebra generated by the ${e \grave{}}_{j}$. Unless explicitly mentioned, we will always assume a superspace with bosonic variables ($m\not=0$). The most important element of the algebra $\cP \otimes \cC$ is the vector variable $x = \ux+\uxb$ with
\[
\ux = \sum_{i=1}^m x_i e_i \,\;,\;\;\;\;\; \uxb= \sum_{j=1}^{2n} {x \grave{}}_{j} {e \grave{}}_{j}.
\]

The square of $x$ is scalar-valued and equals $x^2 =  \sum_{j=1}^n {x\grave{}}_{2j-1} {x\grave{}}_{2j}  -  \sum_{j=1}^m x_j^2=\uxb^2+\ux^2$. We will often write $r=\sqrt{-\ux^2}$. The bosonic part $\ux^2$ is invariant under $SO(m)$ while $\uxb^2$ is invariant under the symplectic group $Sp(2n)$, so $x^2$ is invariant under $SO(m)\times Sp(2n)$.

On the other hand, the super Dirac operator is defined as
\[
\px = \upxb-\upx = 2 \sum_{j=1}^{n} \left( {e \grave{}}_{2j} \partial_{{x\grave{}}_{2j-1}} - {e \grave{}}_{2j-1} \partial_{{x\grave{}}_{2j}}  \right)-\sum_{j=1}^m e_j \pj.
\]

Its square is the super Laplace operator
\[
\Delta = \px^2 =4 \sum_{j=1}^n \partial_{{x \grave{}}_{2j-1}} \partial_{{x \grave{}}_{2j}} -\sum_{j=1}^{m} \pj^2=\Delta_f+\Delta_b.
\]

Usually the Dirac operator acts from the left, but for Cauchy formulas we will also need the right Dirac operator. This is defined by
\[
\begin{array}{ll}
\cdot \px = - \cdot \upxb - \cdot \upx ;& \quad F \rightarrow - F \upxb - F \upx = F\px,
\end{array}
\]
where we have introduced an extra minus sign. The derivatives act from the right in the normal way, $\cdot {x\grave{}}_i \partial_{{x\grave{}}_j}=\cdot \delta_{ij}-\cdot  \partial_{{x\grave{}}_j}{x\grave{}}_i$. If we let $\px$ act on $x$ we obtain
\[
\px x = m-2n = M=(x\px)
\]
where $M$ is the so-called super-dimension. Note that the anti-commuting variables behave as if their dimension is negative. The numerical parameter $M$ gives a global characterization of our superspace and will be very important in the sequel.

When we consider a more general bosonic function space $\cF$ (e.g $\cS(\mR^m)$, $C^{k}(\Omega)$ or $L_2(\mR^m)$), we use the notation $\cF_{m|2n}$ for $\cF\otimes\Lambda_{2n}$.

Furthermore we introduce the super Euler operator by
\begin{eqnarray*}
\mE =\mE_b+\mE_f&=& \sum_{j=1}^m x_j \pj+\sum_{j=1}^{2n} {x \grave{}}_{j} \pjb,
\end{eqnarray*}
which allows us to decompose $\cP$ as
\begin{eqnarray*}
\cP &=& \bigoplus_{k=0}^{\infty} \cP_k, \quad \cP_k=\left\{ p \in \cP \; | \; \mE p=k p \right\}.
\end{eqnarray*}

Now we define spherical harmonics in superspace.

\begin{definition}
An element $F \in \cP$ is a spherical harmonic of degree $k$ if it satisfies
\begin{eqnarray*}
\Delta F &=&0\\
\mE F &=& kF, \quad \mbox{i.e. $F \in \cP_k$}.
\end{eqnarray*}
Moreover the space of all spherical harmonics of degree $k$ is denoted by $\cH_k$.
\end{definition}
In the purely bosonic case we denote $\cH_{k}$ by $\cH_{k}^{b}$, in the purely fermionic case by $\cH_{k}^{f}$.

We have the following decomposition (see \cite{DBS2}).

\begin{lemma}[Fischer decomposition 1]
If $M \not \in -2 \mN$, $\cP$ decomposes as
\begin{equation}
\cP = \bigoplus_{k=0}^{\infty} \cP_k= \bigoplus_{j=0}^{\infty} \bigoplus_{k=0}^{\infty} x^{2j}\cH_k.
\label{superFischer}
\end{equation}
If $m=0$, then the decomposition is given by
\begin{equation}
\Lambda_{2n} = \bigoplus_{k=0}^{n} \left(\bigoplus_{j=0}^{n-k} \uxb^{2j} \cH^f_k \right).
\label{fermFischer}
\end{equation}
\label{scalFischer}
\end{lemma}

The following calculation leads to the Fischer decomposition in lemma \ref{scalFischer}. It will also be needed in the sequel. For $R_k \in \cP_k$ we have that
\begin{equation}
\label{relationslaplace}
\Delta (x^{2t} R_{k})= 2t(2k+M+2t-2) x^{2t-2} R_k + x^{2t} \Delta R_k.
\end{equation}

The notion of a spherical harmonic can be refined to

\begin{definition}
An element $F \in \cP\otimes\cC$ is a spherical monogenic of degree $k$ if it satisfies
\begin{eqnarray*}
\px F &=&0\\
\mE F &=& kF, \quad \mbox{i.e. $F \in \cP_k$}.
\end{eqnarray*}
Moreover the space of all spherical monogenics of degree $k$ is denoted by $\cM_k$.
\end{definition}
In the purely bosonic case we denote $\cM_{k}$ by $\cM_{k}^{b}$, in the purely fermionic case by $\cM_{k}^{f}$.

It is clear that every spherical monogenic is a spherical harmonic. This allows us to refine the Fischer decomposition, leading to (see \cite{DBS2})

\begin{lemma}[Fischer decomposition 2]
If $M \not \in -2 \mN$, $\cP_k\otimes\cC$ decomposes as
\begin{equation*}
\cP_k\otimes\cC  = \bigoplus_{i=0}^{k} x^{i} \cM_{k-i}.
\end{equation*}
If $m=0$, then the decomposition is given by
\begin{eqnarray}
\label{cliffordfermionicfischer}
\cP_k\otimes\cW_{2n}  &=& \bigoplus_{i=0}^{k} \uxb^{\, i} \cM_{k-i}^f, \quad k\leq n\\
\label{cliffordfermfischer2}
\cP_{2n-k} \otimes\cW_{2n}&=& \bigoplus_{i=0}^{k} \uxb^{2n-2k+i} \cM_{k-i}^f, \quad k\leq n.
\end{eqnarray}
\label{cliffordfischerdecomp}
\end{lemma}

For the purely fermionic case we thus obtain the full decomposition 
\begin{eqnarray}
\label{fulferFischer}
\Lambda_{2n}\otimes\cW_{2n}&=&\bigoplus_{k=0}^n\bigoplus_{j=0}^{2n-2k}\uxb^j\cM_k^f.
\end{eqnarray}

In general a function $f\in C^{1}(\Omega)_{m|2n}$ with $\Omega$ an open domain in $\mR^m$ is called left-monogenic in $\Omega$ if $\px f=0$.

By analogy with the purely bosonic (see \cite{MR1169463}) case we define the Gamma operator as
\begin{equation}
\label{EulerGamma}
\Gamma=x\px-\mE.
\end{equation}
Unlike the Euler operator, the Gamma operator is not the sum of the purely bosonic and fermionic Gamma operator. The super Laplace-Beltrami operator is defined by
\begin{equation}
\label{GammaLB}
\Delta_{LB}=\Gamma(M-2-\Gamma ),
\end{equation}
and satisfies 
\begin{equation}
\label{defLB}
x^2\Delta=\Delta_{LB}+\mE(M-2+\mE).
\end{equation}

The integration used on $\Lambda_{2n}$ is the so-called Berezin integral (see \cite{MR732126,DBS5}), defined by
\[
\int_{B} = \pi^{-n} \partial_{{x \grave{}}_{2n}} \ldots \partial_{{x \grave{}}_{1}} = \frac{(-1)^n \pi^{-n}}{4^n n!} \Delta_f^{n}.
\]

Suppose $R_{2n-2k}$ is an element of $\Lambda_{2n}$ with $\mE R_{2n-2k}=(2n-2k)R_{2n-2k}$. Then, as was obtained in \cite{DBSCauchy},
\begin{eqnarray}
\label{berekeningBer}
\int_B\uxb^{2k} R_{2n-2k} &=& \frac{k!  (-1)^{n-k} }{\pi^n4^{n-k}(n-k)!} \Delta_f^{n-k}( R_{2n-2k}).
\end{eqnarray}

The following property of the Berezin integral is easily derived.

\begin{lemma}
\label{basicprop}
If for $p\in\Lambda_{2n}$ holds that 
\begin{eqnarray*}
\int_B p\, q&=&0
\end{eqnarray*}
for every $q\in\Lambda_{2n}$, then $p=0$.
\end{lemma}

The integration used on a general superspace is defined by
\begin{equation}
\label{superint}
\int_{\mR^{m | 2n}} = \int_{\mR^m} dV(\ux)\int_B=\int_B \int_{\mR^m} dV(\ux),
\end{equation}
with $dV(\ux)$ the usual Lebesgue measure in $\mR^{m}$.
In \cite{DBS5} and \cite{DBE1} integration over the supersphere for polynomials was introduced using the following formula
\begin{equation}
\label{Pizzetti}
\int_{SS} R  =  \sum_{k=0}^{\infty} (-1)^k \frac{2 \pi^{M/2}}{2^{2k} k!\Gamma(k+M/2)} (\Delta^{k} R )(0).
\end{equation}
This is based on an old result of Pizzetti (see \cite{PIZZETTI}) concerning integration of polynomials over the ordinary unit sphere. The properties of this integral are listed below, and when $M\not\in -2\mN$, this defines the integral uniquely (see \cite{DBE1}).
\begin{theorem}
If $M \not \in -2\mN$, the only linear functional $T: \cP \rightarrow \mR$ satisfying the following properties for all $f(x) \in \cP$:
\begin{itemize}
\item $T(x^2 f(x)) = - T(f(x))$
\item $T(f(g \cdot x)) = T(f(x))$, \quad $\forall g \in SO(m)\times Sp(2n)$
\item $k \neq l \quad \Longrightarrow \quad T(\cH_k \cH_l) = 0 = T(\cH_l \cH_k)$, i.e. $H_{k} \, \bot \, H_{l}$
\item $T(1) = \dfrac{2 \pi^{M/2}}{\Gamma(M/2)}$,
\end{itemize}
is given by the Pizzetti integral (formula (\ref{Pizzetti})).
\label{uniciteitgeval}
\end{theorem}

We will prove that this unicity also holds in case $M\in-2\mN$ in section \ref{nietstand}.

If $R_{k} \in \cP_{k}$, we have the following connection between integration over the supersphere and integration over the entire superspace
\begin{equation}
\label{RmnSS}
\int_{\mR^{m|2n}}R_k\exp(x^2)=\frac{1}{2}\Gamma (\frac{k+M}{2})\int_{SS}R_k=\int_{0}^\infty v^{k+M-1}\exp(-v^2)dv\int_{SS}R_k.
\end{equation}
The last expression only holds in case the integral is finite $(k+M>0)$. When $M\in-2\mN$, the gamma function in equation (\ref{RmnSS}) can become infinite, but this is compensated by the gamma function in equation (\ref{Pizzetti}). 

We repeat some important facts about spherical harmonics in superspace when $m\not=0$. The proofs can be found in \cite{DBE1}.

\begin{lemma}
If $q < n$ and $k < n-q+1$, there exists a homogeneous polynomial $f_{k,p,q}=f_{k,p,q}(\ux^2,\uxb^2)$ of total degree $k$, unique up to a multiplicative constant, such that $f_{k,p,q} \cH_p^b \otimes \cH_q^f \neq 0$ and
\[
\Delta (f_{k,p,q} \cH_p^b \otimes \cH_q^f) = 0.
\]
The explicit form of this polynomial is given by
\[
 f_{k,p,q}=\sum_{s=0}^k \binom{k}{s}\frac{(n-q-s)!}{\Gamma (\frac{m}{2}+p+k-s)} \ux^{2k-2s}\uxb^{2s}.
\]
\label{polythm}
\end{lemma}

Using these polynomials we can obtain a full decomposition of $\cP$ under the action of $SO(m)\times Sp(2n)$. This follows from lemma \ref{scalFischer} and the following theorem.

\begin{theorem}[Decomposition of $\cH_k$]
Under the action of $SO(m) \times Sp(2n)$ the space $\cH_k$ decomposes as
\label{decompintoirreps}
\begin{equation}
\cH_{k} = \bigoplus_{i=0}^{\min(n,k)} \cH^b_{k-i} \otimes \cH^f_{i} \;\; \oplus \;\; \bigoplus_{j=0}^{\min(n, k-1)-1} \bigoplus_{l=1}^{\min(n-j,\lfloor \frac{k-j}{2} \rfloor)} f_{l,k-2l-j,j} \cH^b_{k-2l-j} \otimes \cH^f_{j},
\end{equation}
with $f_{l,k-2l-j,j}$ the polynomials determined in lemma \ref{polythm}.
\end{theorem}

The orthogonality condition for integration over the supersphere can now be made even stronger. 

\begin{theorem}
\label{integorth}
One has that 
\[
f_{i,p,q} \cH^b_{p} \otimes \cH^f_{q} \quad  \bot \quad f_{j,r,s} \cH^b_{r} \otimes \cH^f_{s}
\]
with respect to the Pizzetti integral if and only if $(i,p,q) \neq (j,r,s)$.
\end{theorem}

In \cite{DBS6} the fundamental solution for the super Dirac operator was calculated,
\begin{eqnarray}
\label{fundopl}
\nu_1^{m|2n} &=& \pi^n \sum_{k=0}^{n-1} 2 \frac{4^k k!}{(n-k-1)!} \nu_{2k+2}^{m|0} \uxb^{2n-2k-1} + \pi^n\sum_{k=0}^n \frac{4^k k! }{(n-k)!} \nu_{2k+1}^{m|0} \uxb^{2n-2k},
\end{eqnarray}
where $ \nu_{j}^{m|0}$ is a fundamental solution of $\upx^j$. This fundamental solution satisfies
\begin{eqnarray*}
\px \nu_1^{m|2n}(x-y)&=&\delta(x-y)=\delta(\ux-\uy)\frac{\pi^n}{n!}(\uxb-\uyb)^{2n}.
\end{eqnarray*}

In \cite{DBS9} the super Fourier transform on $\cS(\mR^m)_{m|2n}$ was introduced as
\begin{eqnarray}
\label{Four}
\cF^{\pm}_{m|2n}(f(x))(y)&=&\int_{\mR^{m|2n},x}\exp(\mp i\langle x,y\rangle)f(x),
\end{eqnarray}
with
\begin{equation}
\label{inprod}
\langle x,y\rangle=\langle\ux,\uy\rangle+\langle \uxb,\uyb\rangle=-\sum_{i=1}^mx_iy_i+\frac{1}{2}\sum_{j=1}^n({x\grave{}}_{2j-1}{y\grave{}}_{2j}-{x\grave{}}_{2j}{y\grave{}}_{2j-1}),
\end{equation}
yielding an $SO(m)\times Sp(2n)$-invariant generalization of the purely bosonic Fourier transform.

We already know from \cite{DBS6} that monogenic functions in superspace are infinitely differentiable. This follows essentially from the fact that all the bosonic parts are polyharmonic. In \cite{DBS2}, the surjectivity of the Dirac operator on the set of super polynomials was proven. Now we generalize this to the set of infinitely differentiable functions defined over an open subset of $\mR^m$. In the proof we will extensively use the decomposition (\ref{fulferFischer}) and that technique will be important for the sequel.

\begin{lemma}
\label{vglpx}
For every $h\in C^{\infty}(\Omega )_{m|2n}\otimes \cC$, with $\Omega$ a open subset of $\mR^m$, there exists a $g\in C^{\infty}(\Omega )_{m|2n}\otimes \cC$, such that
\begin{eqnarray*}
\px g&=& h
\end{eqnarray*}
holds in $\Omega $.
\end{lemma}

\begin{proof}
This lemma holds in purely bosonic Clifford analysis, see \cite{MR697564} page 161. Now, using the Fischer decomposition (\ref{fulferFischer}), we consider a basis of $\Lambda_{2n}\otimes \cW_n$ given by $\{\uxb^jM_k^{l}\}$  with $M_k^{l}$ a basis for the fermionic spherical monogenics of degree $k$ with $0\le k\le n$ and $j+2k\le 2n$. We normalize $M_k^{l,j}=C_{j,k}M_k^l$, for convenience, so that $\upxb \uxb^j M_k^{l,j}=\uxb^{j-1}M_k^{l,j-1}$. This means we can expand a general $g\in C^{\infty}(\Omega )_{m|2n}\otimes \cC$ as $g=\sum_{k=0}^{n}\sum_{j=0}^{2n-2k}\sum_{l}g_{j,k,l}\uxb^jM_k^{l,j}$, with $g_{j,k,l}\in C^{\infty}(\Omega )\otimes \mR_{0,m}$. First we calculate $\px g$, yielding
\begin{eqnarray*}
\px g&=&-\sum_{k=0}^{n}\sum_{j=0}^{2n-2k}\sum_l\upx g_{j,k,l}\uxb^jM_k^{l,j}+\sum_{k=0}^{n}\sum_{j=1}^{2n-2k}\sum_{l}g_{j,k,l}\uxb^{j-1}M_k^{l,j-1}\\
&=&-\sum_{k=0}^{n}\sum_{j=0}^{2n-2k}\sum_l\upx g_{j,k,l}\uxb^jM_k^{l,j}+\sum_{k=0}^{n}\sum_{j=0}^{2n-2k-1}\sum_{l}g_{j+1,k,l}\uxb^{j}M_k^{l,j}\\
&=&-\sum_{k=0}^{n}\sum_l\upx g_{2n-2k,k,l}\uxb^{2n-2k}M_k^{l,2n-2k}+\sum_{k=0}^{n}\sum_{j=0}^{2n-2k-1}\sum_{l}(g_{j+1,k,l}-\upx g_{j,k,l})\uxb^{j}M_k^{l,j}.
\end{eqnarray*}
Similarly we expand $h$ as $h=\sum_{k=0}^{n}\sum_{j=0}^{2n-2k}\sum_{l}h_{j,k,l}\uxb^jkM_k^{l,j}$, with $h_{j,k,l}\in C^{\infty}(\Omega )\otimes \mR_{0,m}$. We see that finding a $g$ for which $\px g=h$ is equivalent with finding $\{ g_{j,k,l}\}$ for which
\begin{eqnarray*}
\upx g_{2n-2k,k,l}&=& -h_{2n-2k,k,l}
\end{eqnarray*}
and
\begin{eqnarray*}
\upx g_{j,k,l}&=& -h_{j,k,l}+g_{j+1,k,l},
\end{eqnarray*}
for $j< 2n-2k$. This can be done by using the purely bosonic theorem in \cite{MR697564} and by induction on $k$.
\end{proof}

\section{Bosonic, fermionic and super Cauchy formulas}
\label{bosfercauch}

An interesting feature of Clifford analysis is that it allows for the construction of several Cauchy-type formulae in higher dimensions (see e.g. \cite{MR1026856, MR1249888, MR1012510}). We start with repeating some well-known facts about bosonic Cauchy formulas. This is necessary to obtain the properties that the fermionic Cauchy formulas will have to satisfy.

\subsection{Bosonic Cauchy formulas}

Starting from Stokes' theorem and using $d(fd\ux^{m-1}g)=(f\upx)d\ux^{m}g+fd\ux^{m}(\upx g)$ with $d\ux=\sum_{j=1}^me_jdx_j$, we find the bosonic Cauchy formula (see \cite{MR697564})
\begin{eqnarray*}
\label{genCauchybosA}
\int_{\partial \Sigma} fd\ux^{m-1}g&=&\int_{\Sigma}(f\upx)d\ux^mg+fd\ux^m(\upx g),
\end{eqnarray*}
where $\Sigma$ is a compact oriented $m$-dimensional submanifold of $\mR^{m}$ and $\partial \Sigma$ its boundary.

After multiplying with the pseudoscalar $e_1\ldots e_m$ we obtain the vector-valued surface element $\ds=\sum_{j=1..m}(-1)^je_jdx_1\ldots \widehat{dx_j}\ldots dx_m$ and the volume-element $dV(\ux)=dx_1\ldots dx_m$, yielding
\begin{eqnarray}
\label{genCauchybos}
\int_{\partial \Sigma} f\ds g&=&\int_{\Sigma}[(f\upx)g+f(\upx g)]dV(\ux).
\end{eqnarray}
This formula is what we call a dimensionally correct Cauchy formula, connecting the integral of functions over the boundary of a manifold with the integral of their derivatives over the entire manifold. It is moreover easy to see that this formula is a direct consequence of its one-dimensional analogue
\begin{eqnarray}
\label{1int}
\int_I f'&=&\int_{\partial I} f,
\end{eqnarray}
with $I$ an interval and $\partial I$ its boundary consisting of two points.

Suppose that the boundary of $\Sigma$ is determined by an equation $\nu(\ux)=0$ and moreover $\nu(\ux)>0$ if $\ux\in\mathring{\Sigma}$ and $\nu(\ux) <0$ if $\ux\in \mR^m\backslash \Sigma$. Then we can rewrite the Cauchy formula (\ref{genCauchybos}) as\begin{equation}
\label{boscauch}
\int_\Sigma [(f\upx)g+f(\upx g)]\, dV(\ux)=-\int_{\partial \Sigma}f(\upx\nu)g\, dS(\ux),
\end{equation}
with $dS$ the elementary surface measure. 

Even more general, if we consider a distribution $\alpha$ with compact support and if $f$ and $g$ are $C^\infty$ functions, then 
\begin{eqnarray*}
\sum_{j=1}^m\int_{\mR^m}\partial_{x_j}(fe_j\alpha g)\, dV(\ux)&=&0.
\end{eqnarray*}
This can be rewritten as
\begin{equation}
\label{boscauch2}
\int_{\mR^m} [(f\upx)\alpha g+f \alpha (\upx g)]\, dV(\ux)=-\int_{\mR^m}f\upx(\alpha) g\,dV(\ux),
\end{equation}
which is the most general form of a Cauchy formula in $\mR^{m}$.

The case $\alpha=H(\nu)$, with $H$ the Heaviside function and $\nu(\ux)$ as above corresponds to the Cauchy formula in equation (\ref{boscauch}). Equation (\ref{boscauch2}) is more general. For example, $\alpha$ could be a Dirac distribution so we would find derivatives of Dirac distributions in the formulas. Such a formula no longer fulfills the condition that we integrate a function $f$ over the boundary of the domain and integrate $\upx f$ over the entire domain. In the case of the Dirac distribution we would integrate twice over the boundary. For that reason we will call such formulas non dimensionally correct Cauchy formulas. They will be necessary to obtain Cauchy formulas related to the supersphere.

\subsection{Fermionic Cauchy formulas}

The Berezin integral (see \cite{MR732126}) of a function $g\in\mC\{1,{x\grave{}}_1\}=\Lambda_1$ of one fermionic variable ${x\grave{}}_1$ is defined by
\begin{equation*}
\pi\int_Bf=\partial_{{x\grave{}}_1}g=(\partial_{{x\grave{}}_1}g)(0).
\end{equation*}
This can formally be written in a similar form as the bosonic equation (\ref{1int})
\begin{eqnarray}
\label{defBer}
\int_0g'&=&\int_{\mR^{-1}}g.
\end{eqnarray}

The one dimensional interval $I$ is replaced by the zero dimensional point $\{ 0\}$ and the zero dimensional boundary $\partial I$ by the minus one dimensional $\mR^{-1}$. In this view $\partial \{0\}=\mR^{-1}$, as was already argued in \cite{MR1402921}. So the rewritten definition of Berezin integration (\ref{defBer}) is the starting point for obtaining a Cauchy formula, exactly as the one dimensional bosonic case (\ref{1int}). 

To obtain a general Cauchy formula in superspace, we take inspiration from the distributional form in equation (\ref{boscauch2}). Instead of distributions $\alpha$, we now take for $\alpha$ again an element of the Grassmann algebra, because the dual space of the finite dimensional vector space $\Lambda_{2n}$ is of course $\Lambda_{2n}$.

We then obtain the following fermionic Cauchy formula. Note that with $f\widehat{\alpha}\upxb$ we mean the fermionic Dirac operator acting from the left on $f\alpha$ but $\alpha$ is not derived. We cannot switch $\alpha$ and $\upxb$ from place because of the anticommuting variables.

\begin{lemma}[fermionic Cauchy formula]
\label{fercauch}
For $f$, $g\in\Lambda_{2n}\otimes\cW_{2n}$ and $\alpha \in\Lambda_{2n}$, the following holds
\begin{eqnarray*}
-\int_B(f\widehat{\alpha}\upxb) g+\int_B f\alpha(\upxb  g)&=&\int_B f(\alpha\upxb ) g.
\end{eqnarray*}
\end{lemma}
\begin{proof}
It suffices to prove the lemma for $f$ and $g\in\Lambda_{2n}$. From the fact that $\int_B\partial_{{x\grave{}}_j}=0$ we deduce that
\begin{eqnarray*}
-\int_B(f\widehat{\alpha}\partial_{{x\grave{}}_j}) g+\int_B f\alpha(\partial_{{x\grave{}}_j}  g)-\int_B f(\alpha\partial_{{x\grave{}}_j} ) g&=&0.
\end{eqnarray*}
Multiplying with the symplectic Clifford algebra generators and summing yields the desired result.
\end{proof}

\begin{remark}
Like the bosonic Cauchy formula in equation (\ref{genCauchybos}), these Cauchy formulas are `dimensionally correct'. In equation (\ref{genCauchybos}), integration over $\Sigma$ of $\upx f$ and integration over $\partial \Sigma$ of $f$ are connected. In lemma \ref{fercauch} the role of the characteristic function of $\Sigma$ is played by $\alpha $ and $\alpha\upxb$ leads to a formal integration with one dimension lower.
\end{remark}

With these Cauchy formulas we can prove a Morera's theorem for the Grassmann algebra. We will repeat the bosonic Morera's theorem in section \ref{dimcorr}, but for now the important consequence of this theorem is that we have `enough' fermionic Cauchy formulas. This was not the case with the formula obtained in \cite{DBSCauchy}.

\begin{lemma}[Morera's theorem in Grassmann algebras]
If for every $\alpha\in\Lambda_{2n}$ it holds that
\begin{eqnarray*}
\int_B(\alpha\upxb)f&=&0
\end{eqnarray*}
for an $f\in\Lambda_{2n}\otimes\cW_{n}$ then $f$ is left monogenic.
\end{lemma}

\begin{proof}
Using lemma \ref{fercauch} this leads to $\int_B\alpha(\upxb f)=0$ for every $\alpha$. This means that $\upxb f=0$ by lemma \ref{basicprop}.
\end{proof}

\subsection{Cauchy formulas in superspace}

Combining the bosonic and fermionic Cauchy formulas we obtain the following Cauchy formula on superspace.

\begin{theorem}
\label{Cauchy}
For $f$ and $g$ $\in C^\infty(\Omega)_{m|2n}\otimes\cC$ and $\alpha\in \mathcal{E}'\otimes\Lambda_{2n} $ a distribution with compact support $\Sigma\subset \Omega$ one has
\begin{eqnarray*}
\int_{\mR^{m|2n}}\left[(f\widehat{\alpha} \px) g+f\alpha(\px g)\right]dV(\ux)=-\int_{\mR^{m|2n}}f(\alpha\px ) g\, dV(\ux).
\end{eqnarray*}
\end{theorem}

\begin{proof}
This follows from equation (\ref{boscauch2}) and lemma \ref{fercauch} for an $\alpha=\beta\gamma$ with $\beta\in\mathcal{E}'$ and $\gamma\in\Lambda_{2n}$.
\end{proof}

Using this formula we can create two kinds of Cauchy formulas. There are Cauchy formulas where $\alpha$ contains only Heaviside distributions corresponding to characteristic functions of subspaces of $\mR^m$. This generalizes the Cauchy formula obtained in \cite{DBSCauchy} and is related to the version of Stokes' formula in \cite{MR1402921}. Indeed, they connect integration of functions over a domain with a certain superdimension with integration of the Dirac operator acting on these functions over a domain with this superdimension plus one. Therefore we call them dimensionally correct Cauchy formulas. In that case, we have the following corollary.

\begin{corollary}
\label{veralgCauchy}
For $\beta\in\Lambda_{2n}$ and $f,g\in C^1(\Sigma)_{m|2n} \otimes \cC$, with $\Sigma$ a compact oriented differentiable $m$-dimensional manifold with smooth boundary $\partial \Sigma$, the following holds
\begin{eqnarray*}
\int_{\Sigma}\int_B[(f\widehat{\beta}\px) g+f\beta(\px g)]dV(\ux)&=&-\int_{\partial \Sigma}\int_Bf\beta d\sigma_{\ux}g+\int_{\Sigma}\int_Bf(\beta\upxb)g dV(\ux).
\end{eqnarray*}
\end{corollary}
\begin{proof}
This is a special case of theorem \ref{Cauchy} with $\alpha=H(\nu)\beta$ with $\nu(\ux)>0$ if $\ux\in\mathring{\Sigma}$, $\nu <0$ if $\ux\in \mR^m\backslash \Sigma$. It also follows immediately from equation (\ref{genCauchybos}) combined  with lemma \ref{fercauch}.
\end{proof}
In this corollary, as in \cite{DBSCauchy, MR1402921}, the boundary of the domain consists of two parts. The two parts of the boundary are determined by $(\partial \Sigma,\beta)$, which is the classical boundary and $(\Sigma,\beta\upxb)$ where the $\beta\upxb$ leads to a dimension which is $1$ lower than the domain ($\Sigma,\beta$). The Cauchy formula in \cite{DBSCauchy} is a special case of corollary \ref{veralgCauchy}, with $\beta = \exp(\uxb^2)-1$ and $\upxb \beta = - \beta \upxb = 2 \uxb \exp(\uxb^2)$.

To properly generalize certain bosonic Cauchy formulas to superspace, $\alpha$ will have to contain also more general distributions. These Cauchy formulas will not have the dimensional property of purely bosonic or fermionic Cauchy formulas (equation (\ref{genCauchybos}) or lemma \ref{fercauch}). This dimensional property is a consequence of the combination of the Heaviside and Dirac distribution. The non-dimensionally correct Cauchy formulas will generalize the integration over the unit sphere $\mS^{m-1}$ to the so-called supersphere. This integration was introduced in \cite{DBS5} and \cite{DBE1}. In the bosonic case the connection between the unit sphere and unit ball yields of course a real (dimensionally correct) Cauchy formula.

The fact that theorem \ref{Cauchy} leads to a unification of the dimensionally correct Cauchy formulas and the formulas for the supersphere and superball shows the power of our approach.

\section{Dimensionally correct Cauchy formulas}
\label{dimcorr}

In this section we prove that corollary \ref{veralgCauchy} leads to Morera's theorem in superspace as well as a Cauchy-Pompeiu formula.

\subsection{Morera's theorem}

In complex analysis, Morera's theorem states that if $f$ is a continuous, complex-valued function defined on an open set $\Omega$ in the complex plane, satisfying $ \oint_T f(z)\,dz = 0$, for every closed triangle $T$ in $\Omega$, then $f$ is holomorphic in $\Omega$. A generalization of this theorem also exists for monogenic functions in Clifford analysis (see \cite{MR697564}, page 60):

\begin{theorem}
\label{classmorera}
A function $f$ is left monogenic in the open set $\Omega\subset \mR^m$ if and only if $f$ is continuous in $\Omega$ and 
\begin{eqnarray*}
\int_{\partial I}d\sigma_{\ux} f&=&0
\end{eqnarray*}
for all intervals $I\subset\Omega$. 
\end{theorem}

Before proving Morera's theorem in superspace we need the following lemma, which is an extension of theorem \ref{classmorera}.

\begin{lemma}
\label{bosmorera2}
If $f\in C^0(\Omega)\otimes\mR_{0,m}$ and $h\in C^{\infty}(\Omega)\otimes\mR_{0,m}$ (with $\Omega$ an open subset of $\mR^m$) and for every interval $I\subset\Omega$
\begin{eqnarray}
\label{morerafh}
\int_{\partial I}d\sigma_{\ux}f&=&\int_{I}hdV(\ux),
\end{eqnarray}
then $\upx f=h$ in $\Omega$.
\end{lemma}

\begin{proof}
Because $h\in C^{\infty}(\Omega)\otimes\mR_{0,m}$ there exists a $g\in C^{\infty}\otimes\mR_{0,m}$ such that $h=\upx g$ (using the surjectivity of the bosonic Dirac operator). Using equations (\ref{genCauchybos}) and (\ref{morerafh}) we find
\begin{eqnarray*}
\int_{\partial I}d\sigma_{\ux}[f-g]&=&0\;\;\;\; \forall I.
\end{eqnarray*}

Using theorem \ref{classmorera} yields $\upx f=\upx g=h$.
\end{proof}

The reverse of this lemma is also true because of equation (\ref{genCauchybosA}). Now we are ready to prove Morera's theorem in superspace, giving a solution to problem \textbf{P3}.

\begin{theorem}[Morera's theorem in superspace]
\label{morera}
A function $f\in C^0(\Omega)_{m|2n}\otimes \cC$ (with $\Omega$ an open subset of $\mR^m$) is left monogenic in $\Omega$ if and only if
\[
\int_{\partial I}\int_B\alpha d\sigma_{\ux}f - \int_{I}\int_B(\alpha\upxb)fdV(\ux)= 0,
\]
for every interval $I\subset\Omega$ and for every $\alpha$ in $\Lambda_{2n}$.
\end{theorem}

\begin{proof}
If $f$ is left monogenic then the integral equality follows immediately from corollary \ref{veralgCauchy}. Now suppose $f\in C^0(\Omega)_{m|2n}\otimes \cC$. The proposed integral equation can be transformed, using lemma \ref{fercauch}, to
\begin{eqnarray*}
\int_{\partial I}\int_B\alpha d\sigma_{\ux}f&=&\int_{I}\int_B\alpha(\upxb f)dV(\ux).
\end{eqnarray*}

This equation holds for every $\alpha$. So by lemma \ref{basicprop}, this means that we have
\begin{eqnarray}
\label{moreratss}
\int_{\partial I} d\sigma_{\ux}f&=&\int_{I}(\upxb f)dV(\ux)\;\;\;\;\; \forall I.
\end{eqnarray}

Now we take for $\Lambda_{2n}\otimes \cW_{2n}$ the same basis $\{\uxb^kM_j^{l,k}\}$ as used in the proof of lemma \ref{vglpx}. So we expand $f$ as $f=\sum_{k=0}^{n}\sum_{j=0}^{2n-2k}\sum_{l}f_{j,k,l}\uxb^jM_k^{l,j}$. The condition on the integrals (\ref{moreratss}) splits into (we take $k$ and $l$ as fixed but arbitrary now)
\begin{eqnarray}
\label{inductiemorera}
\int_{\partial I} d\sigma_{\ux}f_{j-1,k,l}&=&\int_{I}f_{j,k,l}dV(\ux)\;\;\;\;\forall j=1, \ldots, 2n-2k\;\;\;\;\forall I
\end{eqnarray}
and
\begin{eqnarray*}
\int_{\partial I} d\sigma_{\ux}f_{2n-2k,k,l}&=&0\;\;\;\;\forall I.
\end{eqnarray*}

Using theorem \ref{classmorera} we immediately find that $f_{2n-2k,k,l}$ is monogenic in $\Omega$, which also implies that $f_{2n-2k,k,l}\in C^\infty(\Omega)$. Now we proceed by induction (from $j= 2n-2k-1$ to $j=0$), we suppose that $\upx f_{j,k,l}= f_{j+1,k,l}$ and $f_{j,k,l}$ is polyharmonic (and therefore $C^\infty$). This means, by lemma \ref{bosmorera2} and equation (\ref{inductiemorera}), that $\upx f_{j-1,k,l}= f_{j,k,l}$. This also implies that $f_{j,k-1,l}$ is polyharmonic. So, finally, we have found that $f$ is differentiable and that
\begin{eqnarray*}
\px f&=&-\sum_{k=0}^{n}\sum_{j=0}^{2n-2k-1}\uxb^j\sum_lM_k^{l,j}\upx f_{j,k,l}+\sum_{k=0}^{n}\sum_{j=1}^{2n-2k}\sum_{l}\uxb^{j-1}M_k^{l,j-1} f_{j,k,l}
\end{eqnarray*}
will be zero.
\end{proof}

\begin{remark}
In the light of the link between fermionic and bosonic Cauchy formulas, the different $\alpha$ taken in theorem \ref{morera} correspond to the different intervals $I$ taken in theorem \ref{classmorera}. 
\end{remark}

In superspace we can also prove the extended version (lemma \ref{bosmorera2}) of Morera's theorem.

\begin{lemma}
\label{morera2}
If $f\in C^0(\Omega)_{m|2n}\otimes\cC$ and $h\in C^{\infty}_{m|2n}\otimes\cC$ (with $\Omega$ an open subset of $\mR^m$) and for every interval $I\subset\Omega$ and for every $\alpha\in\Lambda_{2n}$,
\begin{eqnarray*}
-\int_{\partial I}\int_B\alpha d\sigma_{\ux}f+\int_{I}\int_B(\alpha\upxb)fdV(\ux)&=&\int_{I}\int_B \alpha hdV(\ux)
\end{eqnarray*}
then $\px f=h$ in $\Omega$.
\end{lemma}

\begin{proof}
The proof is similar to the proof of lemma \ref{bosmorera2} and follows from lemma \ref{vglpx} and theorem \ref{morera}.
\end{proof}
\subsection{Cauchy-Pompeiu formula}

Using the Cauchy kernel (\ref{fundopl}) we can prove the following Cauchy-Pompeiu theorem in superspace. This was previously impossible with the result obtained in \cite{DBSCauchy}. We hence have solved problem \textbf{P4}.

\begin{theorem}[Cauchy-Pompeiu]
\label{CP}
Let $\Sigma\subset\Omega\subset\mR^m$ be a compact oriented differentiable m-dimensional manifold with smooth boundary. Let $g\in C^1(\Omega)_{m|2n}$ and let $\nu_1^{m|2n}$ be the fundamental solution of the super Dirac operator. Then one has
\[
\int_{\partial\Sigma}\int_B\nu_1^{m|2n}(x-y)\ds g(x)+\int_{\Sigma}\int_B\nu_1^{m|2n}(x-y)(\px g(x))dV(\ux)
\]

\[
=
\begin{cases}
0 & \mbox{if}\qquad \uy \in \Omega \backslash \Sigma \\
-g(y) & \mbox{if} \qquad \uy \in\mathring{ \Sigma}
\end{cases}
\]
\end{theorem}

\begin{proof}
The proof of this theorem is exactly the same as the proof of theorem 4 in \cite{DBSCauchy}. There the result is proven for $\beta=\exp(\uxb^2/2)-1$ in corollary \ref{veralgCauchy}, while here we use $\beta=1$.
\end{proof}

\section{Integration over the supersphere}
\label{superspheresection}

In this section we work with general superdimensions $M$ (but $m\not=0$). The case $m=1$ poses some minor problems because then the bosonic unit ball $\mB^m$ and sphere $\mS^{m-1}$ reduce to the interval $[-1,1]$ and its boundary. The theorems we will find will also hold for this case, but with a different normalization.

\subsection{The standard integration over the supersphere}
\label{superspheresection1}

In this section we define integration over the supersphere using a super-analogue of the bosonic formula
\begin{eqnarray*}
\int_{\mS^{m-1}}f|_{r=1}&=&2\int_{\mR^m}\delta(\ux^2+1)f.
\end{eqnarray*}
The integration over the supersphere that we will obtain will be an extension of the already known Pizzetti integration over the supersphere for polynomials (see formula (\ref{Pizzetti})). First, we define the distributions corresponding to the supersphere and superball. 
\begin{definition}
\label{superDirac}
The Heaviside and Dirac distribution corresponding to the supersphere are defined as the Taylor series in $\theta$ of $f(\ux^2+\theta)$ modulo $\theta^{n+1}$ with the identification $\theta=\uxb^2$. The Heaviside distribution which can be seen as the characteristic function of the superball is given by
\[
H(x^2+1)=H(\ux^2+1)+\sum_{j=1}^n\frac{\uxb^{2j}}{j!}\delta^{(j-1)}(\ux^2+1).
\]
The Dirac distribution corresponding to the supersphere is given by
\[
\delta(x^2+1)=\sum_{j=0}^n\frac{\uxb^{2j}}{j!}\delta^{(j)}(\ux^2+1).
\]
\end{definition}

Using these definitions we define integration over the supersphere and superball. We use the same notation as in (\ref{Pizzetti}) and will prove in theorem \ref{intPiz} that this is correct.

\begin{definition}
\label{defdef}
The integral over the supersphere of a function $f\in C^n(\Omega)_{m|2n}$ with $\Omega$ an open domain, $\mS^{m-1}\subset \Omega\subset \mR^m$, is given by
\begin{eqnarray*}
\int_{SS} f&=&2\int_{\mR^{m|2n}}\delta(x^2+1) f.
\end{eqnarray*}
The integral over the superball of a function $f\in C^{n-1}(\Omega)_{m|2n}$ with $\Omega$ an open domain, $\mB^m\subset \Omega\subset \mR^m$, is given by
\begin{eqnarray*}
\int_{SB} f&=&\int_{\mR^{m|2n}}H(x^2+1) f.
\end{eqnarray*}
\end{definition}

We generalize the integrals in definition \ref{defdef} to superspheres of radius $R$, by defining
\begin{equation*}
\int_{SS(R)} f=2R \int_{R^{m|2n}}\delta(x^2+R^2) f \qquad ,\qquad \int_{SB(R)}f=\int_{R^{m|2n}}H(x^2+R^2) f.
\end{equation*}

We calculate the integral over the supersphere with radius $R$ explicitely. We will obtain a closed form without distributions, so this form can be used as a definition for integration over the supersphere without any ambiguity.

\begin{lemma}
\label{uitgewerktint}
For $f\in C^n(\Omega)_{m|2n}$ with $\Omega$ an open domain such that $\mS^{m-1}(R)\subset \Omega\subset \mR^m$, one has
\begin{eqnarray*}
\int_{SS(R)}f&=&\sum_{j=0}^n\left[\int_{\mS^{m-1}}(\frac{\partial}{\partial r}\frac{1}{2r})^jr^{m-1}\int_B\frac{\uxb^{2j}}{j!}f\right]_{r=R}\\
&=&R \sum_{j=0}^n\int_{\mS^{m-1}}\int_B\frac{\uxb^{2j}}{j!}\left[(\frac{\partial}{\partial r^2})^jr^{m-2}f\right]_{r=R}.
\end{eqnarray*}
\end{lemma}

\begin{proof}
Substituting definition \ref{superDirac} in definition \ref{defdef} yields
\begin{eqnarray*}
\int_{SS(R)} f&=&2R\sum_{j=0}^n\int_{\mR^{m}}\delta^{(j)}(\ux^2+R^2)\int_B\frac{\uxb^{2j}}{j!}f\\
&=&2R\sum_{j=0}^n\int_{\mS^{m-1}}\int_{\mR}dr[(-\frac{1}{2r}\frac{\partial}{\partial r})^j\delta(R^2-r^2)]r^{m-1}\int_B\frac{\uxb^{2j}}{j!}f\\
&=&2R\sum_{j=0}^n\int_{\mS^{m-1}}\int_{\mR}dr\,\frac{\delta(r-R)}{2R}(\frac{\partial}{\partial r}\frac{1}{2r})^jr^{m-1}\int_B\frac{\uxb^{2j}}{j!}f.
\end{eqnarray*}
\end{proof}

Now we prove that our integration over the supersphere reduces to the Pizzetti formula (\ref{Pizzetti}) when integrating polynomials. This solves problem \textbf{P1}.
\begin{theorem}
\label{intPiz}
The integration in definition \ref{defdef}
\begin{equation*}
\int_{SS}f=\sum_{j=0}^n\int_{\mS^{m-1}}\int_B\frac{\uxb^{2j}}{j!}\left[(\frac{\partial}{\partial r^2})^jr^{m-2}f\right]_{r=1}
\end{equation*}
is an extension of the Pizzetti integral over the supersphere for polynomials
\begin{equation*}
\int_{SS}R=\sum_{k=0}^{\infty} (-1)^k \frac{2 \pi^{M/2}}{2^{2k} k!\Gamma(k+M/2)} (\Delta^{k} R )(0).
\end{equation*}
\end{theorem}

\begin{proof}
We prove the statement for $P_{p}Q_{q}$ with $P_p$ a bosonic homogeneous polynomial of degree $p$ and $Q_q$ a homogeneous element of $\Lambda_{2n}$ of degree $q$. For ease of notation we will in general write $\Gamma(k+1)/\Gamma(k-l+1)=(\partial_{u})^lu^k$, even when $l<k$, instead of $0$. We have that
\[
2\int_{R^{m|2n}}P_pQ_q\delta(x^2+1)=\sum_{j=0}^n\int_{\mS^{m-1}}(\frac{\partial}{\partial r^2})^jr^{m-2+p}P_p(\uxi)\int_B\frac{\uxb^{2j}}{j!}Q_q(\uxb)|_{r=1}.
\]
Due to symmetry, the integral of a homogeneous polynomial of odd degree over the unit sphere is zero. Similarly, the Berezin integral also vanishes unless $q$ is even. So we can restrict ourselves to the case $p=2k$, $q = 2l$ and we are reduced to calculating
\begin{eqnarray*}
2\int_{R^{m|2n}}P_{2k}Q_{2l}\delta(x^2+1)&=&[\int_{\mS^{m-1}}P_{2k}(\uxi)][\int_B\frac{\uxb^{2(n-l)}}{(n-l)!}Q_{2l}(\uxb)](\frac{\partial}{\partial {r^2}})^{n-l}r^{m-2+2k}|_{r=1}.
\end{eqnarray*}

Using equation (\ref{Pizzetti}) in the purely bosonic case and equation (\ref{berekeningBer}) we find
\begin{eqnarray*}
2\int_{R^{m|2n}}P_{2k}Q_{2l}\delta(x^2+1)&=&(-1)^k\frac{2\pi^{m/2}}{2^{2k}k!\Gamma(k+m/2)}\Delta_b^k(P_{2k})\;\frac{(-1)^l\pi^{-n}}{2^{2l}l!}\Delta_f^l(Q_{2l})\;\frac{\Gamma(\frac{m}{2}+k)}{\Gamma(\frac{m}{2}+k-n+l)}\\
&=&(-1)^{k+l}\frac{2\pi^{M/2}}{2^{2(k+l)}(k+l)!\Gamma(k+l+M/2)}\Delta^{k+l}(P_{2k}Q_{2l})
\end{eqnarray*}
which equals $\int_{SS}P_{2k}Q_{2l}$ in equation (\ref{Pizzetti}). This completes the proof, as also for the super Pizzetti formula only even degrees of bosonic and fermionic monomials will give a non-zero contribution. 
\end{proof}

By exactly the same methods we can prove that the integral with $H(x^2+1)$ is an extension of the integral over the superball from \cite{DBS5}, yielding
\begin{eqnarray*}
\int_{SB}f&=&\int_{\mB^m}\int_Bf+\sum_{j=0}^{n-1}\int_{\mS^{m-1}}\int_B\frac{\uxb^{2j+2}}{(j+1)!}\left[(\frac{\partial}{\partial r}\frac{1}{2r})^jr^{m-1}f\right]_{r=1}\\
&=&\sum_{k=0}^{\infty} (-1)^k \frac{\pi^{M/2}}{2^{2k} k!\Gamma(k+M/2+1)} (\Delta^{k} f)(0)
\end{eqnarray*}
for $f\in\cP$.

The following equality follows from a straightforward calculation.

\begin{lemma}
\label{2keerR}
For $f\in C^n(\mR^m)_{m|2n}$ and $R\in\mR^+$.
\begin{eqnarray}
\int_{SS(R)}f(x)&=&R^{M-1}\int_{SS}f(Rx).
\end{eqnarray}
\end{lemma}

Now we can prove equation (\ref{RmnSS}) for general functions.

\begin{theorem}
\label{connectieBer}
For $f$ a function in $L_1(\mR^m)_{m|2n}\cap C^n(\mR^m)_{m|2n}$ and $M>0$, the following relation holds
\begin{eqnarray*}
\int_{\mR^{m|2n}}f&=&\int_{0}^\infty dR\, R^{M-1}\int_{SS,x}f(Rx).
\end{eqnarray*}
\end{theorem}

\begin{proof}
We calculate the right-hand side using lemma \ref{2keerR} and lemma \ref{uitgewerktint}.
\begin{eqnarray*}
\int_{0}^\infty dR\, R^{M-1}\int_{SS,x}f(Rx)&=&\int_{0}^\infty dR\,\int_{SS(R)}f\\
&=&\int_{0}^\infty dR\sum_{j=0}^n\left[\int_{\mS^{m-1}}(\frac{\partial}{\partial r}\frac{1}{2r})^jr^{m-1}\int_B\frac{\uxb^{2j}}{j!}f\right]_{r=R}\\
&=&\int_{0}^\infty dR\sum_{j=0}^n(\frac{\partial}{\partial R}\frac{1}{2R})^j R^{m-1}\int_{\mS^{m-1}}\int_B\frac{\uxb^{2j}}{j!}f(R\uxi,\uxb)\\
&=&\int_{0}^\infty dR\,R^{m-1}\int_{\mS^{m-1}}\int_Bf(R\uxi,\uxb)\\
&=&\int_{\mR^{m|2n}}f.
\end{eqnarray*}
The terms for $j>0$ are zero because of partial integration and because $\lim_{R\to\infty}R^k(\frac{\partial}{\partial R})^lf=0$ for $k\le m-1$ and $\lim_{R\to 0}R^k(\frac{\partial}{\partial R})^lf=0$ for $0<M\le k$.
\end{proof}

\subsection{Green's theorem on the superball}

With the integrations over the supersphere and superball introduced in section \ref{superspheresection1} we can reconstruct the Green's theorem obtained in \cite{DBS5}. This way, we can generalize this theorem which previously only held for polynomials. We will prove this Green's theorem using the super Cauchy formula from theorem \ref{Cauchy}, which moreover makes the proof a lot shorter than in \cite{DBS5}. In this way we also find a link between the Green's theorem on the superball and the super Cauchy fomulas which was formerly unknown. This connection is clear in the bosonic case. Indeed, when we substitute $\alpha=H(\ux^2+1)$ in the bosonic Cauchy formula (\ref{boscauch2}), we find
\begin{equation*}
\int_{\mR^m} [(f\upx)H(\ux^2+1) g+fH(\ux^2+1) (\upx g)]dV(\ux)=2\int_{\mR^m}f\ux\delta(\ux^2+1) g\,dV(\ux),
\end{equation*}
which is equivalent to Green's theorem on the unit ball:
\begin{equation}
\label{bosSS}
\int_{\mB^m} [(f\upx)g+f(\upx g)]dV(\ux)=\int_{\mS^{m-1}}f\ux g\,dS(\ux).
\end{equation}

Now we generalize this to superspace. We first prove the following lemma
\begin{lemma}
\label{lemsuperDirac}
The Dirac operator acting on the Heaviside distribution yields
\begin{eqnarray*}
\px H(x^2+1) &=& 2x \delta(x^2+1)
\end{eqnarray*}
in distributional sense.
\end{lemma}

\begin{proof}
We calculate that in distributional sense
\begin{eqnarray*}
\px H(x^2+1)&=&-\upx\sum_{j=0}^nH^{(j)}(\ux^2+1)\frac{\uxb^{2j}}{j!}+\upxb\sum_{j=1}^nH^{(j)}(\uxb^2+1)\frac{\uxb^{2j}}{j!}\\
&=&2\ux\sum_{j=0}^nH^{(j+1)}(\ux^2+1)\frac{\uxb^{2j}}{j!}+2\uxb\sum_{j=1}^nH^{(j)}(1+\uxb^2)\frac{\uxb^{2j-2}}{(j-1)!}\\
&=&2\ux\sum_{j=0}^nH^{(j+1)}(\ux^2+1)\frac{\uxb^{2j}}{j!}+2\uxb\sum_{j=0}^nH^{(j+1)}(1+\uxb^2)\frac{\uxb^{2j}}{j!}\\
&=&2xH^{(1)}(x^2+1)\\
&=&2x\delta(x^2+1).
\end{eqnarray*}
\end{proof}

We then obtain the connection between Cauchy's formula and Green's theorem in superspace

\begin{theorem}
\label{Cauchysphere1}
For $f$ and $g$ in $C^n(\Omega)_{m|2n}\otimes \cC$, with $\Omega$ open and $\mB^{m}\subset \Omega\subset \mR^{m}$, the following holds
\begin{eqnarray*}
\int_{SB}[(f\px)g+f(\px g)]&=&-\int_{SS}fx g.
\end{eqnarray*}
\end{theorem}

\begin{proof}
This theorem is  a special case of theorem \ref{Cauchy} using $\alpha=H(x^2+1)$, lemma \ref{lemsuperDirac} and taking into account that only $n$ derivatives are necessary.
\end{proof}

\begin{remark} 
Comparing theorem \ref{Cauchysphere1} with equation (\ref{bosSS}) we see that $x$ behaves as the outer normal on the supersphere.
\end{remark}

Also for the non dimensionally correct Cauchy formulae it is possible to construct an associated Cauchy-Pompeiu formula. We demonstrate this for the superball.

\begin{theorem}[Cauchy-Pompeiu on the supersphere]
For $f\in C^1(\Omega)_{m|2n}$ with $\mB^m\subset \Omega$ the following holds
\begin{equation*}
\int_{SS}\nu_1^{m|2n}(x-y)xf(x)+\int_{SB}\nu_1^{m|2n}(x-y)\px f(x) 
\\=
\begin{cases}
-f(y)&\uy\in\mathring{\mB^m}\\
0&\uy\in\Omega\backslash \mB^m
\end{cases}
\end{equation*}
\end{theorem}
\begin{proof}
For this theorem it is useful to split the Cauchy formula in theorem \ref{Cauchysphere1} into a dimensionally correct part and the remainder, by putting
\begin{eqnarray*}
H(x^2+1)&=&H(\ux^2+1)+C(x).
\end{eqnarray*}

Using lemma \ref{lemsuperDirac} we obtain that $2x\delta(x^2+1)=2\ux\delta(\ux^2+1)+\px C(x)$. Now the distribution $C(x)$ has as support $\mS^{m-1}$. Because $\px\nu_1^{m|2n}(x-y)=0$ everywhere for $\ux\in\mS^{m-1}$ (as $\uy\not\in\mS^{m-1}$), theorem \ref{Cauchy} for $\alpha=C(x)$ implies
\begin{eqnarray*}
\int_{\mR^{m|2n}}\nu_1^{m|2n}(x-y)C(x)(\px g)=-\int_{\mR^{m|2n}}\nu_1^{m|2n}(x-y)(C(x)\px)g.
\end{eqnarray*}
When we rewrite theorem \ref{CP} for $\Sigma=\mB^m$ we find
\begin{equation*}
\int_{\mR^{m|2n}}\nu_1^{m|2n}(x-y)2\ux\delta(\ux^2+1)g(x)+\int_{\mR^{m|2n}}\nu_1^{m|2n}(x-y)H(\ux^2+1)(\px g(x))
\\=\begin{cases}
-g(y)&\uy\in\mathring{\mB^m}\\
0&\uy\in\Omega\backslash \mB^m
\end{cases}
\end{equation*}
Putting these two equations together we find the theorem.
\end{proof}

With the following proposition we can easily re-derive the fact that spherical harmonics of different degree are orthogonal with respect to integration over the supersphere (part of theorem \ref{integorth}). This result was obtained in \cite{DBS5} for the case $M \not \in -2\mN$ and for polynomial functions.
\begin{proposition}
\label{Green2}
For $f$ and $g$ in $C^{n+1}(\Omega)_{m|2n}$ and $\mB^m\subset \Omega$, the following relation holds
\begin{eqnarray*}
\int_{SS}(f\mE g-(\mE f)g)&=&\int_{SB}(f\Delta g-(\Delta f)g).
\end{eqnarray*}
\end{proposition}

\begin{proof}
Using theorem \ref{Cauchysphere1} we find
\begin{eqnarray*}
\int_{SB}(\Delta f)g+(f\px)(\px g)&=&-\int_{SS}(f\px)xg
\end{eqnarray*}
and
\begin{eqnarray*}
\int_{SB} (f\px)(\px g)+f(\Delta g)&=&-\int_{SS}fx(\px g).
\end{eqnarray*}
If we subtract these two equations we find
\begin{eqnarray*}
\int_{SS}fx(\px g)-(f\px)xg&=&\int_{SB}(f\Delta g-(\Delta f)g).
\end{eqnarray*}
The theorem now follows from equation (\ref{EulerGamma}) and the fact that $\mE$ and $\Delta$ are scalar operators, while $\Gamma$ is not.
\end{proof}

Finally, we prove that integration over the supersphere behaves as expected with respect to the action of $\Gamma$.

\begin{proposition}
For $f\in C^{n+2}(\Omega)_{m|2n}\otimes\cC$ with $\Omega$ open and $\mS^{m-1}\subset\Omega\subset \mR^m$ one has
\label{LB0}
\begin{eqnarray*}
\int_{SS}\Gamma f&=&0 \qquad \mbox{and} \qquad \int_{SS}\Delta_{LB} f=0.
\end{eqnarray*}
\end{proposition}

\begin{proof}
It is sufficient to prove the lemma for $f\in C^{n+2}(\Omega)_{m|2n}$. Using equation (\ref{EulerGamma}) and  theorem \ref{Cauchysphere1} we find
\begin{eqnarray*}
\int_{SS}\Gamma f+\mE f&=&\int_{SS}x\px f\\
&=&-\int_{SB}\px\px f\\
&=&-\int_{SB}\Delta f.
\end{eqnarray*}
The first part of the proposition now follows from the fact that $\mE$ and $\Delta$ are scalar while $\Gamma$ is not. The second part follows immediately from equation (\ref{GammaLB}).
\end{proof}

\label{greenball}

\subsection{Other integrations over the supersphere}
\label{nietstand}

In this section we show how the integration over the supersphere we introduced in section \ref{superspheresection1} can be uniquely defined as a functional on $\cP$.

\begin{definition}
The space of integrations over the supersphere $I_{m|2n}$ is the space of all linear functionals $\phi:\cP\to\mR$, with the property that for every $f(x)\in\cP$ 
\begin{equation}
\label{vwint}
\begin{cases}
i.&\phi(x^2f(x))=-\phi(f(x))\\
ii.&\phi(f(g.x))=\phi(f(x))\, ,\;\;\forall g\in\,SO(m)\times Sp(2n).
\end{cases}
\end{equation}
\end{definition}

In \cite{DBE1} it was proven that for $M\not\in-2\mN$, $I_{m|2n}$ is $n+1$ dimensional. Theorem \ref{uniciteitgeval} states that the property $\phi(H_kH_l)=0$ when $k\not= l$, then leaves a unique (up to a constant) integration on superspace. This is given by the Pizzetti formula (\ref{Pizzetti}). It was also proven in \cite{DBE1} that for $M\not\in-2\mN$ there is a unique basis of $I_{m|2n}$, $\{ \int_{SS,j},j=0..n\}$ such that $\int_{SS,0}=\int_{SS}$ and
\begin{eqnarray}
\label{bijkomendebepint}
\int_{SS,j}f_{i,0,0}=\delta_{ij}\;\;\;1\le i,j\le n,
\end{eqnarray}
with $f_{i,0,0}$ from lemma \ref{polythm} and theorem \ref{decompintoirreps}. We can easily see that the functional
\begin{eqnarray*}
\psi_j:P(x)& \to &\int_{SS}f_{j,0,0} P(x)
\end{eqnarray*}
satisfies conditions (\ref{vwint}).  Now, $\int_{SS}f_{i,0,0}f_{j,00}=0$ for $j\not= i$ (theorem \ref{integorth}) and $\int_{SS}f_{i,0,0}f_{i,00}\not=0$ because otherwise $\int_{SS}f_{i,0,0}=\int_{SS}$, see equation (\ref{bijkomendebepint}). So we find $\psi_j=\int_{SS,j}$ up to a constant.

Now in the case $M\in-2\mN$ we can still define integration over the supersphere by the Pizzetti formula (\ref{Pizzetti}). This still has the properties (\ref{vwint}) and the spherical harmonics of different degree are still orthogonal, see \cite{DBS5} (this also follows immediately from proposition \ref{Green2}). It is clear that in that case the first $-M/2$ terms in the summation (\ref{Pizzetti}) are zero. In the following we will search all the functionals which satisfy (\ref{vwint}) and prove that $\int_{SS}$ is still the only one with the property that spherical harmonics of different degree are orthogonal. The methods that we will use are also valid in the case $M\not\in-2\mN$, so the results hold for every superdimension with $m\not= 0$.

\begin{lemma}
\label{dimensiespace}
$I_{m|2n}$ is a finite-dimensional vectorspace of dimension $n+1$.
\end{lemma}
\begin{proof}
In the case $M\in -2\mN$ there is no Fischer decomposition (\ref{superFischer}). We can, however, still decompose the space $\cP$ into irreducible pieces under the action $SO(m)\times Sp(2n)$ using the purely bosonic (special case of (\ref{superFischer})) and purely fermionic Fischer decomposition in equation (\ref{fermFischer}). This yields
\begin{eqnarray*}
\cP&=&\bigoplus_{k=0}^\infty\bigoplus_{i=0}^\infty\bigoplus_{l=0}^n\bigoplus_{j=0}^{n-k} \ux^{2i}\uxb^{2j}\cH_k^b\cH^f_{l}.
\end{eqnarray*}

 By Schur's lemma and (\ref{vwint}$ii.$) the only summands in the Fischer decompositions that give nonzero contributions to an integration over the supersphere are one-dimensional pieces. So an integration over the supersphere is uniquely determined by its values on the algebra generated by $\{\ux^2,\uxb^2\}$. Now taking into account condition (\ref{vwint}$i.$) we see that an integration over the supersphere is determined entirely by its values on $\{1,\uxb^2,\ldots,\uxb^{2n}\}$. These values can be chosen arbitrarily.
\end{proof}

The proof of lemma \ref{dimensiespace} shows that for an arbitrary element $\alpha$ of the algebra $\mbox{Alg}(\uxb^2)$, generated by $\uxb^2$, there exists an integration over the supersphere $\phi$ for which $\phi(\alpha)\not=0$. We now establish a new basis of $I_{m|2n}$.

\begin{lemma}
\label{spaceFunInt}
The space $I_{m|2n}$ has a basis $\{\phi_k|k=0,\ldots ,n\}$ such that
\[\phi_k(P)=\int_{SS}\uxb^{2k}P\]
for $P\in\cP$.
\end{lemma}

\begin{proof}
Because of lemma \ref{dimensiespace} we only have to prove that the $\phi_k$ are $n+1$ linearly independent functionals which satisfy the conditions (\ref{vwint}). Only the linear independence is non-trivial. So suppose 
\begin{eqnarray*}
\sum_{k=l}^na_k\phi_k&=&0
\end{eqnarray*}
with $a_l\not=0$. This means that
\begin{eqnarray*}
\phi_l\left[(\sum_{k=0}^{n-l}a_{k+l}\uxb^{2k})P\right]&=&0
\end{eqnarray*}
for all $P\in\cP$. Because $\sum_{k=0}^{n-l}a_{k+l}\uxb^{2k}$ is invertible in $\Lambda_{2n}$ we obtain $\phi_l=0$. So we still have to prove that none of the functionals in the proposed basis is zero. First we calculate, using theorem \ref{intPiz},
\begin{eqnarray*}
\int_{SS}\uxb^{2n}&=&\int_{\mS^{m-1}}\int_B\uxb^{2n}\\
&=&\sigma_m\pi^{-n}n!\\
&\not=&0.
\end{eqnarray*}

It follows that $\phi_k(\uxb^{2n-2k})\not=0$ and hence $\phi_k\not=0$.
\end{proof}

$\phi_0$ is of course the usual Pizzetti integral over the supersphere. Another interesting case is $\phi_n$. Using definition \ref{defdef} we find that this integration takes the form
\begin{eqnarray*}
\phi_n(P)&=&\int_{\mS^{m-1}}\int_B\uxb^{2n}P\\
&=&\frac{n!}{\pi^n}\int_{\mS^{m-1}}P_0
\end{eqnarray*}
where $P_0$ denotes the purely bosonic part of $P$. So this integration is just the integration of the bosonic part over the bosonic unit sphere. Using theorem \ref{intPiz} we can also define the extension of $\phi_k$ from $\cP$ to the space of functions that are $n-k$ times differentiable. It is also possible to connect the integrations over the supersphere with a Berezin integration over the entire superspace. Using the proof of theorem \ref{connectieBer} we immediately find that
\begin{eqnarray*}
\int_{\mR^{m|2n}}\uxb^{2k}f&=&\int_0^\infty dR\, R^{M+2k-1}\phi_k[f(Rx)],
\end{eqnarray*}
for $M+2k>0$ and an $f\in L_1(\mR^m)_{m|2n}\cap C^n(\mR^m)_{m|2n}$.

Now we will prove that also in the case $M\in-2\mN$ the orthogonality of spherical harmonics of different degree uniquely determines $\int_{SS}$ in the space of integrations over the supersphere. First we need some lemmata.

\begin{lemma}
\label{nieuwef}
For a superdimension $M=-2t$ with $t\in\mN$, one has
\[
f_{k,t+1,0}\not=0\mod x^2,
\]
with $f_{k,t+1,0}$ as defined in lemma \ref{polythm}, for $1\le k \le n$.
\end{lemma}

\begin{proof}
If $f_{k,t+1,0}=0\mod x^2$, there would exist a $g_{2k-2i}\in\cP_{2k-2i}$, $0\le i\le k$ so that 
\begin{eqnarray*}
f_{k,t+1,0}&=&x^{2i}g_{2k-2i}
\end{eqnarray*}

with $g_{2k-2i}\not=0\mod x^2$. Now, because of equation (\ref{relationslaplace}) and lemma \ref{polythm} we find for a $H^b_{t+1}\in\cH^b_{t+1}$,
\begin{eqnarray*}
\Delta f_{k,t+1,0}H^b_{t+1}&=&0\\
&=&2i(4k-4i+2t+2+M+2i-2)x^{2i-2}g_{2k-2i}H_{t+1}^b+x^{2i}\Delta (g_{2k-2i}H_{t+1}^b)\\
&=&2i(4k-2i)x^{2i-2}g_{2k-2i}H_{t+1}^b+x^{2i}\Delta (g_{2k-2i}H_{t+1}^b).
\end{eqnarray*}

Now because $4k-2i\ge 2k>0$ we find that up to a nonzero constant $g_{2k-2i}H_{t+1}^b \sim x^{2}\Delta (g_{2k-2i}H_{t+1}^b)$. So $g_{2k-2i}$ would have a factor $x^2$ which is contradictory to the assumption that $g_{2k-2i}\not=0\mod x^2$.
\end{proof}

\begin{remark}
\label{fkto}
It is clear from the proof that the lemma can be made more general, $f_{k,t+1,0}\not=0\mod x^2$ as long as $M\ge -2t$. So in particular for $M>0$, $f_{k,t,0}\not=0\mod x^2$ for every $t$. By a similar calculation the same also holds for $M$ odd. The fact that $f_{k,0,0}=0\mod x^2$ when $M\in-2\mN$ is the reason why in that case (\ref{bijkomendebepint}) is not a basis for $I_{m|2n}$.
\end{remark}

Now we establish yet another basis for $I_{m|2n}$, suitable for the case $M \in -2\mN$.
\begin{lemma}
\label{nieuwebasisSS}
The functionals $\int_{SS}$ and
\begin{eqnarray*}
\int_{SS}f_{k,t+1,0}\;\;\;1\le k\le n
\end{eqnarray*}
form a basis of $I_{m|2n}$ if $M=-2t$.
\end{lemma}

\begin{proof}
By (\ref{vwint}$i.$), $\int_{SS}f_{k,t+1,0}(\ux^2,\uxb^2)\cdot=\int_{SS}f_{k,t+1,0}(-1-\uxb^2,\uxb^2)\cdot$. Now, suppose \[f_{k,t+1,0}(-1-\uxb^2,\uxb^2)=\sum_{s=0}^ka_s(-1-\uxb^2)^{k-s}\uxb^{2s}\] is of degree $<k$ in $\uxb^2$. This means that
\begin{eqnarray*}
\sum_{s=0}^ka_s(-1)^{k-s}&=&0,
\end{eqnarray*}
so
\begin{eqnarray*}
\sum_{s=0}^ka_s(-\uxb^2)^{k-s}\uxb^{2s}&=&0\\
&=&f_{k,t+1,0}(-\uxb^2,\uxb^2).
\end{eqnarray*}
This would mean that $f_{k,t+1,0}=0\mod x^2$, which leads to a contradiction by lemma \ref{nieuwef}. So $f_{k,t+1,0}(-1-\uxb^2,\uxb^2)$ contains a factor $\uxb^{2k}$. This means we can easily show that going from the basis in lemma \ref{spaceFunInt} to the proposed basis is an invertible transform.
\end{proof}

\begin{lemma}
For $M=-2t$, $H^b_{t+1}\in\cH^b_{t+1}$ a bosonic spherical harmonic of degree $t+1$ and $f_{j,t+1,0}$ as defined in lemma \ref{polythm} the following holds
\label{ffnot0}
\begin{eqnarray*}
\int_{SS}f_{j,t+1,0}H^b_{t+1}\,f_{k,t+1,0}H^b_{t+1}&=&\delta_{jk}C_k,
\end{eqnarray*}
with $C_k\not=0$ for $1\le k\le n$.
\end{lemma}

\begin{proof}
The factor $\delta_{jk}$ follows from theorem \ref{integorth}. Because $f_{k,t+1,0}(-1-\uxb^2,\uxb^2)\,(1+\uxb^2)^{t+1}$ is an element of $\mbox{Alg}(\uxb^2)$ there exists an integration over the supersphere $\phi$ for which $\phi(f_{k,t+1,0}(-1-\uxb^2,\uxb^2)\,(1+\uxb^2)^{t+1})\not=0$. This means that at least one of the functionals of lemma \ref{nieuwebasisSS} will not give zero. So we calculate, using (\ref{vwint}$i.$)
\begin{eqnarray*}
\int_{SS}f_{j,t+1,0}f_{k,t+1,0}(-1-\uxb^2,\uxb^2)\,(1+\uxb^2)^{t+1}&=&(-1)^{t+1}\int_{SS}f_{j,t+1,0}f_{k,t+1,0}(\ux^2,\uxb^2)(\ux^2)^{t+1}.
\end{eqnarray*}

Now we normalize the bosonic spherical harmonic such that $\int_{\mS^{m-1}}H^b_{t+1}H^b_{t+1}=1$. Because of definition \ref{defdef} we see that
\begin{eqnarray*}
\int_{SS}f_{j,t+1,0}f_{k,t+1,0}\,(\ux^2)^{t+1}&=&\int_{SS}f_{j,t+1,0}H^b_{t+1}\,f_{k,t+1,0}H^b_{t+1}.
\end{eqnarray*}

This is zero whenever $j\not=k$ by theorem \ref{integorth}. Because there has to be a functional $\phi$ for which this is not zero we find $\int_{SS}f_{j,t+1,0}H^b_{t+1}f_{j,t+1,0}H^b_{t+1}\not=0$.
\end{proof}

Now we have all tools necessary to solve problem \textbf{P2}.
\begin{theorem}
\label{uniek}
For a general superdimension $M$, with $m\not=0$, the integral over the supersphere $\int_{SS}$, defined in definition \ref{defdef}, is the unique element of $I_{m|2n}$ with the property that 
\begin{eqnarray*}
\int_{SS}H_kH_l&=&0
\end{eqnarray*}
when $k\not=l$, for $H_j\in\cH_j$.
\end{theorem} 

\begin{proof}
We only have to prove the theorem for $M\in-2\mN$, although the proof is also valid for every $M$ by remark \ref{fkto}. We know that $\int_{SS}$ satisfies the property $\int_{SS}H_kH_l=0$ from theorem \ref{integorth}. Now, by lemma \ref{nieuwebasisSS} every integration over the supersphere is of the form
\begin{eqnarray*}
\phi&=&\sum_{j=0}^na_j\int_{SS}f_{j,t+1,0}\cdot
\end{eqnarray*}

Because $H^b_{t+1}\in\cH^b_{t+1}\subset\cH_{t+1}$ and $f_{k,t+1,0}H_{t+1}^b\in \cH_{2k+t+1}$, $\phi(H_{t+1}^b\, f_{k,t+1,0}H_{t+1}^b)=C_k a_k$ should be zero. As $C_k$ is not zero by lemma \ref{ffnot0}, we find $a_j=0$ whenever $j\not=0$, so $\phi=a_0\int_{SS}$.
\end{proof}

\section{Radon transform}

The bosonic Radon transform integrates elements of $\cS(\mR^m)$ over hyperplanes. Some properties of this transform can  be found e.g. in \cite{MR1723736} and \cite{MR709591}. The Radon transform of a function $f\in\cS(\mR^m)$ is given by
\begin{eqnarray*}
\cR_{m}(f)(\uy,p)&=&\int_{\mR^m}\delta(\langle \ux , \uy\rangle+p)f(\ux),
\end{eqnarray*}
with $\langle\ux,\uy\rangle$ the inner product between two vectors as defined in (\ref{inprod}). Using the Fourier transform of the Dirac distribution we find that (with $\cF_{m|0}^-$ as in equation (\ref{Four}))
\begin{eqnarray}
\label{FourierRadon}
\cR_{m}(f)(\uy,p)&=&(2\pi)^{m/2-1}\int_{-\infty}^\infty dr\, \exp(ipr)[\cF^-_{m|0}(f)(r\uy)].
\end{eqnarray}

Because of this property, it was hinted in \cite{DBS9}  that the Radon transform on superspace could be introduced using the super Fourier transform (\ref{Four}) on $\cS(\mR^{m})_{m|2n}$ developed there. This has a few shortcomings that we will be able to solve using distributions in superspace. To define the Radon transform in superspace we take again the Taylor expansion of the Dirac distribution, yielding
\begin{eqnarray*}
\delta(\langle x,y \rangle+p)&=&\sum_{j=0}^{2n}\delta^{(j)}(\langle\ux ,\underline{y}\rangle+p)\frac{\langle\uxb,\underline{y}\grave{}\rangle^j}{j!}.
\end{eqnarray*}

We then have the following definition of the super Radon transform:
\begin{definition}
The Radon transform of a function $f\in \cS(\mR^{m})_{m|2n}$ is given by
\begin{eqnarray*}
\cR_{m|2n}(f)(y,p)=\int_{\mR^{m|2n}}\delta(\langle x,y \rangle+p)f(x).
\end{eqnarray*}
\label{defRadon}
\end{definition}

Now we show that this definition is equivalent with the definition given in \cite{DBS9}, hence solving problem \textbf{P5}.

\begin{theorem} 
\label{fourierradon}
For every $f\in \cS(\mR^{m|2n})$ the Radon transform is given by
\begin{eqnarray*}
\cR_{m|2n}(f)(y,p)&=&(2\pi)^{M/2-1}\int_{-\infty}^\infty dr\,\exp(ipr)[\cF^-_{m|2n}(f)(ry)],
\end{eqnarray*}
for every superdimension $M$ with $m\not=0$.
\end{theorem}

\begin{proof}

This theorem is a direct consequence of the following equality of elements of $(\cS(\mR^{m})_{m|2n})' \cong \cS(\mR^m)'\otimes\Lambda_{2n}$,
\begin{eqnarray*}
\delta(\langle x,y\rangle+p)&=&\sum_{j=0}^{2n}\delta^{(j)}(\langle\ux , \underline{y}\rangle+p)\frac{\langle\uxb,\underline{y}\grave{}\rangle^j}{j!}\\
&=&\frac{1}{2\pi}\sum_{j=0}^n\frac{\langle\uxb,\uyb\rangle^j}{j!}\int_{\mR}dr\,(ir)^j\exp(ir(\langle\ux ,\uy\rangle+p))\\
&=&\frac{1}{2\pi}\int_{\mR}dr\,\exp(ir\langle\uxb,\uyb\rangle)\exp(ir(\langle\ux,\uy\rangle+p))\\
&=&\frac{1}{2\pi}\int_{\mR}dr\,\exp(ir(\langle x,y\rangle+p)).
\end{eqnarray*}
\end{proof}

Further properties of the Radon transform in superspace will be obtained in a subsequent paper.

\end{document}